\documentclass[sigconf,natbib=false,nonacm]{acmart}
\usepackage[T1]{fontenc}
\usepackage[utf8]{inputenc}
\usepackage{eso-pic}
\usepackage{mathtools}

\newcommand{\etal}{\emph{et al.\@}}

\mathchardef\mhyphen="2D
\newcommand{\es}{\epsilon}
\newcommand{\append}{\mathrel{\shortparallel}}
\newcommand{\getsappend}{\overset{\raisebox{-1pt}{\tiny$\;\shortparallel$}}{\gets}}
\newcommand{\True}{\texttt{T}}
\newcommand{\False}{\texttt{F}}
\newcommand{\bits}{\{0,1\}}
\newcommand{\NN}{\mathbb{N}}
\DeclarePairedDelimiter\abs{\lvert}{\rvert}

\newcommand{\defineterm}[1]{\textsl{\color{red!40!black}#1}}
\newcommand{\heading}[1]{\medskip\noindent\textsc{#1\@.}}
\newcommand{\figref}[1]{Fig.\@~\ref{#1}}
\newcommand{\Figref}[1]{Figure~\ref{#1}}
\newcommand{\secref}[1]{Sec.\@~\ref{#1}}

\newcommand{\advA}{\mathcal{A}}
\newcommand{\Adv}{\mathbf{Adv}}

\newcommand{\keysp}{\mathcal{K}}
\newcommand{\msgsp}{\mathcal{M}}
\newcommand{\ctxtsp}{\mathcal{C}}
\newcommand{\ct}{\mathit{ct}}
\newcommand{\adsp}{\mathcal{AD}}
\newcommand{\ad}{\mathit{ad}}
\newcommand{\btsp}{\mathcal{B}\mathit{t}}
\newcommand{\bt}{\mathit{bt}}
\newcommand{\enc}{\mathrm{enc}}
\newcommand{\dec}{\mathrm{dec}}

\newcommand{\SAFE}{\mathrm{SAFE}}
\newcommand{\IND}{\mathrm{IND}}
\newcommand{\INT}{\mathrm{INT}}
\newcommand{\safe}{\mathrm{safe}}
\newcommand{\ind}{\mathrm{ind}}
\renewcommand{\int}{\mathrm{int}}
\newcommand{\Odec}{\mathrm{Dec}}

\newcommand{\NEXT}{\mathrm{nxt}}
\newcommand{\IV}{\mathrm{IV}}
\newcommand{\taglen}{\mathrm{taglen}}
\newcommand{\encode}{\mathrm{encode}}
\newcommand{\decode}{\mathrm{decode}}
\newcommand{\mav}{\mathrm{mav}}

\begin{document}
\fancyhead{}

\title{Efficiency Improvements for Encrypt-to-Self}

\author{Jeroen Pijnenburg}
\affiliation{Royal Holloway, University of London, Egham, UK}
\email{{jeroen.pijnenburg.2017}@rhul.ac.uk}
\author{Bertram Poettering}
\orcid{0000-0001-6525-5141}
\affiliation{IBM Research--Zurich, Rüschlikon, Switzerland}
\email{{poe}@zurich.ibm.com}

\begin{abstract}
Recent work by Pijnenburg and Poettering \mbox{(ESORICS'20)} explores the novel cryptographic Encrypt-to-Self primitive that is dedicated to use cases of symmetric encryption where encryptor and decryptor coincide.
The primitive is envisioned to be useful whenever a memory-bounded computing device is required to encrypt some data with the aim of temporarily depositing it on an untrusted storage device.
While the new primitive protects the confidentiality of payloads as much as classic authenticated encryption primitives would do,
it provides considerably better authenticity guarantees:
Specifically,
while classic solutions would completely fail in a context involving user corruptions,
if an encrypt-to-self scheme is used to protect the data,
all ciphertexts and messages fully remain unforgeable.

\quad
To instantiate their encrypt-to-self primitive, Pijnenburg {\etal} propose a mode of operation of the compression function of a hash function,
with a carefully designed encoding function playing the central role in the serialization of the processed message and associated data.
In the present work we revisit the design of this encoding function.
Without questioning its adequacy for securely accomplishing the encrypt-to-self job,
we improve on it from a technical/implementational perspective
by proposing modifications that alleviate certain conditions that would inevitably require implementations to disrespect memory alignment restrictions imposed by the word-wise operation of modern CPUs,
ultimately leading to performance penalties.
Our main contributions are thus to propose an improved encoding function,
to explain why it offers better performance,
and to prove that it provides as much security as its predecessor.
We finally report on our open-source implementation of the encrypt-to-self primitive based on the new encoding function.
\end{abstract}

\maketitle

\AddToShipoutPicture*{\raisebox{26cm}{\hspace{4.6cm}\parbox{13cm}{\normalfont\footnotesize An extended abstract of this article was accepted for presentation at \href{https://www.cysarm.org/}{CYSARM 2020}. This is a pre-publication authors' copy.}}}

\section{Introduction}
\label{sec:introduction}

\heading{Encrypt-to-Self}
Assume a resource constrained computing device like a smartcard or a TPM chip that is required to temporarily or permanently store a record of data that is larger than what would fit into its memory capacity.
If the device is connected to other devices, e.g., to a dedicated storage server,
a nearby solution would be to transfer the data to the latter and retrieve it from there when needed.
In this article we focus on secure solutions for this,
meaning that the storage server is trusted with as little as possible.
In particular,
with a secure solution,
the storage server should not be able to recover any non-trivial information about the data (confidentiality),
nor should it be able to alter or manipulate the data (integrity, authenticity).

At first sight one might come to the conclusion that an immediate solution is implied by authenticated encryption~\cite{CCS:Rogaway02}.
The resource constrained device would generate and hold the key,
and the storage server would see just ciphertexts.
Seemingly, any off-the-shelf authenticated encryption scheme, like AES-GCM~\cite{NIST:SP800-38D} or ChaCha20/Poly1305~\cite{rfc8439} or OCB3~\cite{rfc7253}, would do the job.

Note that in the described scenario,
the party encrypting and decrypting is the same.
This motivated recent work by Pijnenburg and Poettering~(\textbf{PP})
to coin the term Encrypt-to-Self (\textbf{EtS})
for the adequate type of encryption~\cite{ESORICS:PijPoe20,EPRINT:PijPoe20b}.
As PP point out, standard authenticated encryption does \emph{not} manifest a secure solution to the EtS challenge.
The reason for this is the lack of security in case of user corruptions.
A corruption is an attack where the adversary retrieves a copy of the key of an honest user.
Such a condition can result from side-channel analysis, physical inspection, a computer break-in, leaked backup copies, etc.
If standard encryption techniques are used in the EtS setting,
all security is immediately lost in the moment a corruption based attack happens:
The adversary can decrypt all ciphertexts,
and it can create forgeries on arbitrary self-chosen messages simply by encrypting them.
The authors of~\cite{ESORICS:PijPoe20} identify this as an issue,
and argue that satisfactory EtS solutions should not fail that drastically.

This highlights the constant arms race between attackers and defenders: manifesting itself in cryptology with on one side the design of new primitives, protocols and improved security models and the analysis of said designs on the other side.
The EtS primitive is of particular interest with an explosion of cloud adoption in 2020.
In this article we take another look at the recently introduced EtS primitive \cite{ESORICS:PijPoe20} and after careful analysis identify optimizations to improve upon the proposed construction.

\heading{Security of Encrypt-to-Self}
As suggested by PP, authenticated encryption does not represent a secure solution to~EtS.
Investigating and evaluating (better?\@) suitable candidates for EtS requires first identifying what level of security actually can be reached in the EtS setting.
In the following we discuss this,
considering aspects of confidentiality and authenticity separately.
Regarding confidentiality,
it is clear that no EtS candidate whatsoever could protect message contents any better than a general encryption scheme.
The reason is as follows:
If the constrained device transforms a message to a ciphertext,
it does so with the goal of being able to later recover the message from the ciphertext,
using key material that it stores locally.
Now, if the adversary performs a user corruption,
it retrieves a full copy of this key material,
bringing it into the position of being able to recover the message using exactly the same algorithms as the device would.
That is, after a corruption, no confidentiality can remain.
The situation is different for authentication.
We illustrate this by giving two intuitive reasons:
Firstly,
a construction could, in principle, use different keys for authenticating and verifying ciphertexts,
for instance by employing a signature scheme.
After an authentication key is used,
it would be securely erased so that a corruption will not leak it to an adversary.
The verification key would be kept and remain intact until a decryption is conducted.
In this case, clearly, a corruption would not leak enough information to enable the adversary to forge new valid ciphertexts.
Secondly,
while a corruption attack leaks all information stored at a user,
it does not change this information.
Thus, in the EtS setting, the encryptor could keep for itself a small amount of information about the ciphertext,
in such a way that even if this information is leaked it still remains impossible to forge a new ciphertext matching it.
More concretely,
if storing a message corresponds with encrypting it, sending the ciphertext to the storage server, and keeping a hash value of the ciphertext in a local registry,
then forging a ciphertext would necessarily require, even after a corruption, finding a collision for the hash function.
The conclusion drawn in~\cite{ESORICS:PijPoe20} from these and similar examples is that in the EtS setting a level of authentication is reachable that goes beyond what is in reach with standard authenticated encryption.
In~\cite{ESORICS:PijPoe20},
PP~go on and propose a security model that formalizes EtS security in the presence of user corruptions.
The analyses in the present article are conducted with respect to their model.

\heading{Constructions of Encrypt-to-Self}
If plain authenticated encryption is not a sufficient solution for~EtS,
how do we construct a solution that is?
It turns out that the encrypt-then-hash (\textbf{EtH}) construction suggested above is secure in the model of~PP.
Recall that EtH processes the message in two passes:
first an encryption pass is conducted to derive the ciphertext,
then a hashing pass is conducted to derive the tag from the ciphertext.
While this approach is generic and robust,
for requiring two passes it is not very efficient.
Indeed,
the alternative approach explored by PP in~\cite{ESORICS:PijPoe20}
entangles encryption and hashing into one operation.
This results in a substantial saving of computational work.

\begin{figure*}[h]
  \centering
  \includegraphics{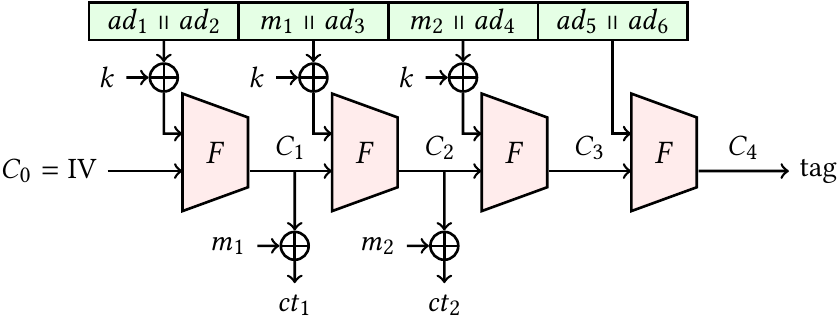}
  \caption[Principle of EtS construction]{\ignorespaces
    Principle of the EtS construction from~\cite{ESORICS:PijPoe20}.
    (Less important details were removed for clarity.)
    An input consisting of a key~$k$, a message~$m$, and associated data~$\ad$,
    is transformed into a ciphertext~$\ct$ and a tag.
    The message is encoded into blocks $m_1,m_2,\ldots$,
    the associated data is encoded into blocks $\ad_1,ad_2,\ldots$,
    and
    the ciphertext consists of the sequence $\ct_1,\ct_2,\ldots$
  \unskip}
  \label{fig:ets:principle}
\end{figure*}

\Figref{fig:ets:principle} provides an overview of (a simplified version of) the EtS scheme from~\cite{ESORICS:PijPoe20}.
\unskip\footnote{\ignorespaces
  A fairly similar construction was proposed by Dodis {\etal}~\cite{C:DGRW18} as a solution to a rather different problem.
\unskip}
Function~$F$ represents the compression function of a Merkle--Damgård (\textbf{MD}) hash function,
that is,
it takes $d\/$~bits of data input and $c$~bits of chain input,
and (deterministically) transforms these inputs to $c$~bits of chain output.
Such compression operations are at the heart of common hash functions like SHA256 and SHA512~\cite{NIST:FIPS180-4},
and are routinely assumed to behave pseudo-randomly.
\unskip\footnote{\ignorespaces
  More formally, they are assumed to behave like a random oracle.
\unskip}
The EtS construction of~\cite{ESORICS:PijPoe20} starts with splitting the message input and optional associated-data input into sequences of blocks $(m_i)$ and~$(\ad_i)$.
It then iteratively computes the MD hash value of these blocks,
\unskip\footnote{\ignorespaces
  In the figure, the intermediate chain values are labelled $C_0,\ldots,C_3$, and the final hash value is labelled~$C_4$.
\unskip}
where all message carrying parts are additionally protected by XORing the EtS key into the corresponding compression function input.
Intuitively, working the key into the chain state in this way ensures that intermediate values $C_1,C_2,\ldots$ are distributed uniformly at random from an adversary's perspective.
This is exploited by employing these values as masks for one-time pad encrypting the message blocks $m_1,m_2,\ldots$ into ciphertext blocks $\ct_1,\ct_2,\ldots$
Overall, this explains how the approach of~\cite{ESORICS:PijPoe20} achieves both confidentiality and authenticity in one pass.

\heading{Encoding for Encrypt-to-Self Construction}
Our description of the EtS scheme of~\cite{ESORICS:PijPoe20} omits an important detail,
namely how precisely the message and associated data is split and formatted as compression function inputs.
Quite obviously,
the encoding has to be injective,
as otherwise it would be easy to find two different pairs $(\ad,m)$ that result in the same tag.
In~\cite{ESORICS:PijPoe20},
PP~propose a suitable and compact yet fairly technical encoding that assumes that the compression function is one-bit \emph{tweakable}~\cite{JC:LisRivWag11}.
That is,
it assumes that the compression function takes a total of $d+c+1$ bits on input
($d\/$~data bits, $c$~chain bits, and $1$~tweak bit),
in order to output $c$~bits of chain state.

While skipping many details of the encoding scheme from~\cite{ESORICS:PijPoe20},
we shed light on two of its properties that are most relevant for the present article.
Firstly,
if a message block~$m$ of length smaller than~$c$ is to be processed,
the block is formatted as $\abs m\append m\append0^*$.
That is, the encoding consists of three concatenated components:
the length of the message (in bytes), the message itself, and a stretch of null bytes so that the desired overall length is reached.
Secondly,
when the key is XORed into the compression function input blocks,
it is actually only XORed into the prefix that contains the message.
For instance,
in the second iteration shown in \figref{fig:ets:principle},
function~$F$ is evaluated on input $(k\oplus m_1)\append\ad_3\append C_1$
rather than on $k\oplus(m_1,\ad_3)\append C_1$.
\unskip\footnote{\ignorespaces
  This notation omits the tweaking bit for clarity of exposition.
\unskip}
As the security argument provided by~PP makes evident,
this is indeed sufficient for security~\cite{ESORICS:PijPoe20}.

\heading{Our Approach}
We propose two improvements for the encoding scheme of~PP~\cite{ESORICS:PijPoe20}.
These are not related to an aspect of security,
but rather to efficiency.
\unskip\footnote{\ignorespaces
  Both the PP scheme and our scheme are provably secure with the same bounds.
  With respect to security, the two schemes are thus equivalent.
  With respect to efficiency, our scheme is superior.
\unskip}
We explain our modifications in detail in \secref{sec:construction},
but we anticipate some details here.
The first modification reconsiders
the encoding $\abs m\append m\append0^*$ of the message (see above).
We note that implementing this will require shifting every message byte in computer memory by one position.
It turns out that due to the word-wise organization of computer memory,
a shift-by-one operation is considerably more expensive that one might assume at first.
\unskip\footnote{\ignorespaces
  Two quantities play a role here:
  Memory move operations should be by multiples of 64~bits (8~bytes) due to the register size of modern CPUs,
  and they should be by multiples of 32~bytes by the size of the cache lines.
\unskip}
We alleviate this efficiency bottleneck by changing the padding to $m\append0^*\append\abs m$
which does not require shifting by a single byte (yet remains injective).
The second modification reconsiders how the key~$k$ is XORed into the compression function inputs.
Here we observe that in most use cases of EtS it should be expected that the associated data string is shorter than the message input.
In the terms of \figref{fig:ets:principle}
this means that the $\ad$-input of most compression function invocations will be constant,
meaning that preparing the compression function input requires just XORing the message with the key.
Our proposal is to switch,
in the terms of the example above,
the XORing step
from $(k\oplus m_1)\append\ad_3\append C_1$
to $(k\oplus\ad_3)\append m_1\append C_1$.
Note that if the associated data input is shorter than the message,
this means that no XORing is necessary once the associated data is fully processed
(as the first component of the concatenation remains invariant and can be precomputed).
Also this modification improves on the execution time of the overall algorithm.

\heading{Our Contributions}
We reconsider the encoding scheme of~PP~\cite{ESORICS:PijPoe20,EPRINT:PijPoe20b}
and suggest alterations as just described.
Our proposals improve the efficiency of the EtS construction of~PP,
rendering it truly practical.
We then formally show that our modified encoding scheme is provably secure.
We finally report on our implementation of the improved EtS scheme.
We will release the source code under a free software license by the time this article goes into print.

\heading{Related Work}
The Encrypt-to-Self primitive emerged only very recently and little directly related work seems to exist.
We already pointed to~\cite{ESORICS:PijPoe20,EPRINT:PijPoe20b}
as our main sources of inspiration,
and to the work of Dodis {\etal}~\cite{C:DGRW18}
for a quite similar solution for a different problem.
In~\cite{ESORICS:PijPoe20},
PP~identify the topics of memory encryption, password managers, and encryptment
(a notion related to instant messaging)
as related to EtS,
although there doesn't really seem to be a considerable overlap.
Finally,
we note that EtS-like tools have been proposed for the state management of TLS~1.3 variants~\cite{EC:AviGelJag19}.

\section{Preliminaries}
\label{sec:prelims}

\subsection{Notation}
\label{sec:preliminaries:notation}

We will keep notation consistent with \cite{ESORICS:PijPoe20} to allow for easy comparison.
We denote the natural numbers with $\NN=\{0,1,\ldots\}$ and we write $\NN^+=\{1,2,\ldots\}$ for the natural numbers excluding zero.
For the Boolean constants True and False we either write $\True$ and~$\False$, respectively, or $1$ and~$0$, respectively,
depending on the context.
An alphabet~$\Sigma$ is any finite set of symbols or characters.
We denote with~$\Sigma^n$ the set of strings of length~$n$
and with $\Sigma^{\leq n}$ the strings of length up to (and including)~$n$.
For our implementation we assume that $\abs\Sigma=256$,
i.e., that all strings are byte strings.
We denote string concatenation with~$\append$.
If $\mathit{var}$ is a string variable and $\mathit{exp}$ evaluates to a string, we write $\mathit{var}\getsappend\mathit{exp}$ shorthand for $\mathit{var}\gets\mathit{var}\append\mathit{exp}$.
Further, if $\mathit{exp}$ evaluates to a string,
we write $\mathit{var}\append\mathit{var}'\gets_n\mathit{exp}$ to denote splitting $\mathit{exp}$ such that we assign the first $n$ characters from $\mathit{exp}$ to $\mathit{var}$ and assign the remainder to $\mathit{var}'$.
When we do not need the remainder,
we write $\mathit{var}\gets_n\mathit{exp}$ to denote assigning the first $n$ characters from $\mathit{\exp}$ to $\mathit{var}$.
In pseudocode, we write $\$(S)$ for picking an element of~$S$ uniformly at random, for any finite set $S$.
Associative arrays implement the `dictionary' data structure:
Once the instruction $A[\cdot]\gets\mathit{exp}$ initialized all items of array~$A$ to the default value~$\mathit{exp}$,
with $A[\mathit{idx}]\gets\mathit{exp}$ and $\mathit{var}\gets A[\mathit{idx}]$
individual items indexed by expression~$\mathit{idx}$ can be updated or extracted.
Finally, we note all algorithms considered in this article may be randomized.

\subsection{Security Games}
\label{sec:preliminaries:games}

The security analyses in our article are conducted with respect to the security model formalized in \cite{ESORICS:PijPoe20}, we replicate the security model for completeness.
Security games are parameterized by an adversary,
and consist of a main game body plus zero or more oracle specifications.
The execution of a game starts with the main game body and terminates when a `Stop with~$\mathit{exp}$' instruction is reached,
where the value of expression~$\mathit{exp}$ is taken as the outcome of the game.
The adversary can query all oracles specified by the game, in any order and any number of times.
If the outcome of a game $\mathrm{G}$ is Boolean,
we write $\Pr[\mathrm{G}(\advA)]$ for the probability that an execution of~$\mathrm{G}$ with adversary~$\advA$ results in True,
where the probability is over the random coins drawn by the game and the adversary.
We define macros for specific combinations of game-ending instructions:
We write `Win' for `Stop with~$\True$' and `Lose' for `Stop with~$\False$',
and further
`Reward~$\mathit{cond}$' for `If $\mathit{cond}$: Win',
`Promise~$\mathit{cond}$' for `If $\lnot\mathit{cond}$: Win',
`Require~$\mathit{cond}$' for `If $\lnot\mathit{cond}$: Lose'.
These macros emphasize the specific semantics of game termination conditions.
For instance, a game may terminate with `Reward~$\mathit{cond}$' in cases where the adversary arranged for a situation
\unskip---indicated by $\mathit{cond}$ resolving to True---\ignorespaces
that should be awarded a win
(e.g., the crafting of a forgery in an authenticity game).

\subsection{Handling of Algorithm Failures}
\label{sec:preliminaries:allalgoscanfail}

We follow the clean notation of \cite{ESORICS:PijPoe20} where \emph{any} algorithm of a cryptographic scheme can fail.
Here, by failure it is meant that an algorithm doesn't generate output according to its syntax specification,
but instead outputs some kind of error indicator
(e.g., an AE decryption algorithm that rejects an unauthentic ciphertext or a randomized signature algorithm that doesn't have sufficiently many random bits to its disposal).
Instead of encoding this explicitly in syntactical constraints which would clutter the notation,
we assume that if an algorithm invokes another algorithm as a subroutine, and the latter fails,
then also the former immediately fails.
\unskip\footnote{\ignorespaces
This approach to handling algorithm failures is taken from~\cite{ToSC:PijPoe20} and borrows from how modern programming languages handle `exceptions',
where any algorithm can raise (or `throw') an exception,
and if the caller does not explicitly `catch' it,
the caller is terminated as well and the exception is passed on to the next level.
See \href{https://en.wikipedia.org/wiki/Exception_handling_syntax}{Wikipedia:\,\texttt{Exception\_handling\_syntax}} for exception handling syntaxes of many different programming languages.
\unskip}
The same is assumed for game oracles:
If an invoked scheme algorithm fails, then the oracle immediately aborts as well.
Further, we assume that the adversary learns about this failure,
i.e., the oracle will return the error indicator when it aborts.
Note that this implies that if a scheme's algorithms leak vital information through error messages, then the scheme will not be secure in these models.
(That is, they are particularly robust.)
We believe that this way to handle errors implicitly rather than explicitly contributes to obtaining definitions with clean and clear semantics.

\subsection{Memory Alignment}
\label{sec:memoryalignment}
For $n$ a power of~2,
we say an address of computer memory is $n$-byte aligned if it is a multiple of $n$~bytes.
We further say that a piece of data is $n$-byte aligned if the address of its first byte is $n$-byte aligned.
A modern CPU accesses a single (aligned) word in memory at a time.
Therefore, the CPU performs reads and writes to memory most efficiently when the data is aligned.
For example, on a 64-bit machine, 8~bytes of data can be read or written with a single memory access if the first byte lies on an 8-byte boundary.
However, if the data does not lie within one word in memory, the processor would need to access two memory words, which is considerably less efficient.
We modify the scheme algorithms proposed by \cite{ESORICS:PijPoe20} such that when they need to move around data, they exclusively do this for aligned addresses.
In practice, the preferred alignment value depends on the hardware used,
so for generality in this article we refer to it abstractly as the \underline{m}emory \underline{a}lignment \underline{v}alue~$\mav$.
(A typical value would be $\mav=8$.)

\section{Notions of Encrypt-to-Self}
\label{sec:encrypttoself}

In \cite{ESORICS:PijPoe20} the authors identified the novel encrypt-to-self (EtS) primitive,
which provides one-time secure encryption with authenticity guarantees that hold beyond key compromise.
In this section we replicate the syntax and security definitions of EtS.

EtS consists of an encryption and a decryption algorithm,
where the former translates a message to a \defineterm{binding tag} and a ciphertext,
and the latter recovers the message from the tag-ciphertext pair.
For versatility the two operations further support the processing of an associated-data input~\cite{CCS:Rogaway02}
which has to be identical for a successful decryption.

The task of the binding tag is to prevent forgery attacks:
A user that holds an authentic copy of the binding tag will never accept any ciphertext they did not generate themselves, even if all their secrets become public.
Note that while standard authenticated encryption (AE) does not provide this type of authentication,
the encrypt-then-hash construction suggested in \secref{sec:introduction} does.
In \secref{sec:construction} we provide a considerably more efficient construction that uses a hash function's compression function as its core building block.

\begin{definition}
Let $\adsp$ be an associated data space
and let $\msgsp$ be a message space.
An \defineterm{encrypt-to-self} {\/\normalfont\/ (EtS)} scheme
for $\adsp$ and $\msgsp$
consists of algorithms $\enc,\dec$,
a key space~$\keysp$, a binding-tag space~$\btsp$, and a ciphertext space~$\ctxtsp$.
The encryption algorithm~$\enc$ takes
a key $k\in\keysp$, associated data $ad\in\adsp$ and a message $m\in\msgsp$,
and returns a binding tag $\bt\in\btsp$ and a ciphertext $c\in\ctxtsp$.
The decryption algorithm~$\dec$ takes
a key $k\in\keysp$, a binding tag $\bt\in\btsp$, associated data $\ad\in\adsp$ and a ciphertext $c\in\ctxtsp$,
and returns a message $m\in\msgsp$.
A shortcut notation for this API is as follows:
\[
  \keysp\times\adsp\times\msgsp\to\enc\to\btsp\times\ctxtsp
\]
and
\[
  \keysp\times\btsp\times\adsp\times\ctxtsp\to\dec\to\msgsp
  \enspace.
\]
\end{definition}

\heading{Correctness and Security}
We require of an EtS scheme that if a message~$m$ is processed to a tag-ciphertext pair with associated data~$\ad$,
and a message~$m'$ is recovered from this pair using the same associated data~$\ad$,
then the messages $m,m'$ shall be identical.
This is formalized via the $\SAFE$ game in \figref{fig:encrypttoself:games}.
\unskip\footnote{\ignorespaces
  The SAFETY term borrows from the Distributed Computing community.
  SAFETY should not be confused with a notion of security.
  Informally, safety properties require that ``bad things'' will not happen.
  (In the case of encryption, it would be a bad thing if the decryption of an encryption yielded the wrong message.)
  For an initial overview we refer to
  \href{https://en.wikipedia.org/wiki/Safety_property}{Wikipedia:\,\texttt{Safety\_property}}
  and for the details to~\cite{DC:AlpSch87}.
\unskip}
In particular, observe that if the adversary queries $\Odec(\ad,c)$
(for the authentic $\ad$ and~$c$ that it receives in line~02)
and the $\dec$ procedure produces output~$m'$, the game promises that $m'=m$
(lines~05,06).
Recall from \secref{sec:preliminaries:games} that this means the game stops with output~$\True$ if $m'\neq m$.
Intuitively, the scheme is \defineterm{safe}
if we can rely on $m'=m$, that is,
if the maximum advantage
$\Adv^\safe(\advA)\coloneqq\max_{\ad\in\adsp,m\in\msgsp}\Pr[\SAFE(\ad,m,\advA)]$
that can be attained by realistic adversaries~$\advA$ is negligible.
The scheme is perfectly safe if $\Adv^\safe(\advA)=0$ for all~$\advA$.
We remark that the universal quantification over all pairs $(\ad,m)$ makes the advantage definition particularly robust.

\begin{figure*}[t]
  \centering
  {\scalebox{1.045}{\includegraphics{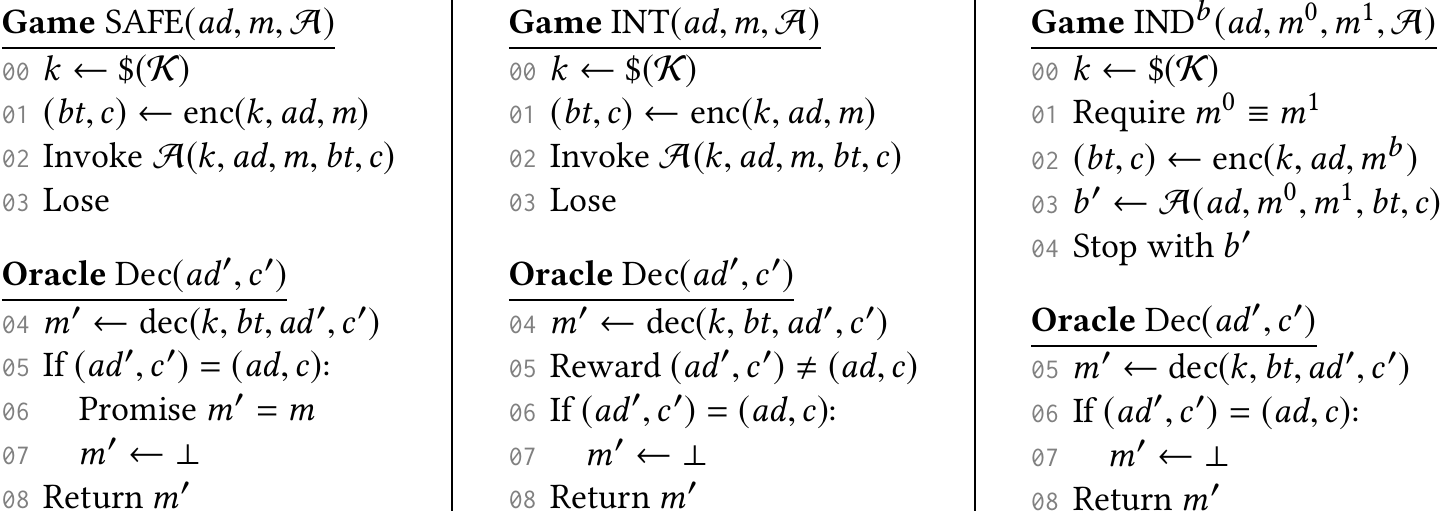}}}
  \caption[Games for EtS]{\ignorespaces
    Games for EtS.
    For the values~$\ad',c'$ provided by the adversary we require that $\ad'\in\adsp,c'\in\ctxtsp$.
    Assuming $\bot\notin\msgsp$,
    we encode suppressed messages with~$\bot$.
    For the meaning of instructions Stop with, Lose, Promise, Reward, and Require see \secref{sec:preliminaries:games}.
  \unskip}
  \label{fig:encrypttoself:games}
\end{figure*}

The security notions demand that the integrity of ciphertexts be protected (INT-CTXT),
and that encryptions be indistinguishable in the presence of chosen-ciphertext attacks (IND-CCA).
The notions are formalized via the $\INT$ and $\IND^0,\IND^1$ games in \figref{fig:encrypttoself:games},
where the latter two depend on some equivalence relation
${\equiv}\subseteq\msgsp\times\msgsp$ on the message space.
\unskip\footnote{\ignorespaces
  \label{fn:ind-uses-equiv}\ignorespaces
  We use relation~$\equiv$ (in line~01 of~$\IND^b$)
  to deal with certain restrictions that practical EtS schemes may feature.
  Concretely, our construction does not take effort to hide the length of encrypted messages,
  implying that indistinguishability is necessarily limited to same-length messages.
  In the formalization this technical restriction is expressed by defining~$\equiv$
  such that $m^0\equiv m^1 \Leftrightarrow \abs{m^0}=\abs{m^1}$.
\unskip}
For consistency,
in line~07
in each of the three games
we suppress the message
if the adversary queries $\Odec(\ad,c)$.
This is crucial in the $\IND^b$ games,
as otherwise the adversary would trivially learn which message was encrypted,
but does not harm in the other games as the adversary already knows~$m$.
Recall from \secref{sec:preliminaries:allalgoscanfail} that all algorithms can fail, and if they do, then the oracles immediately abort.
This property is crucial in the $\INT$ game where the $\dec$ algorithm must fail for unauthentic input such that the oracle immediately aborts.
Otherwise, the game will reward the adversary, that is the game stops with~$\True$
(line~05).
We say that a scheme provides \defineterm{integrity} if the maximum advantage
$\Adv^\int(\advA)\coloneqq\max_{\ad\in\adsp,m\in\msgsp}\Pr[\INT(\ad,m,\advA)]$
that can be attained by realistic adversaries~$\advA$ is negligible,
and that it provides \defineterm{indistinguishability} if the same holds for the advantage
\begin{align*}
    &\Adv^\ind(\advA)\coloneqq\\
    &\max_{\substack{\ad\in\adsp \\ m^0,m^1\in\msgsp}}
    \abs{\Pr[\IND^1(\ad,m^0,m^1,\advA)]-\Pr[\IND^0(\ad,m^0,m^1,\advA)]}.
\end{align*}

\section{New Encrypt-to-Self Construction}
\label{sec:construction}
We mentioned in \secref{sec:introduction} that a generic construction of EtS can be realized by combining standard symmetric encryption with a cryptographic hash function:
one encrypts the message and computes the binding tag as the hash of the ciphertext.
In \cite{ESORICS:PijPoe20} the authors provide a more efficient construction that builds on the compression function of a Merkle--Damgård hash function.
To be more precise, the construction uses a tweakable compression function with tweak space $T=\bits$,
i.e., the domain of the compression function is extended by one bit.
We provide a general definition below.

\begin{definition}\label{def:CF}
    For $\Sigma$ an alphabet, $c,d\in\NN^+$ with $c\leq d$, and a tweak space~$T$,
    we define a \defineterm{tweakable compression function}
    to be a function $F\colon\Sigma^d\times T\times\Sigma^c\to\Sigma^c$
    that takes as input a block $B\in\Sigma^d$ from the data domain,
    a tweak $t\in T$ from the tweak space,
    and a string $C\in\Sigma^c$ from the chain domain,
    and outputs a string $C'\in\Sigma^c$ in the chain domain.
\end{definition}

We will write $F_t(B,C)$ as shorthand notation for $F(B,t,C)$.
For practical tweakable compression functions the memory alignment value~$\mav$
(see \secref{sec:memoryalignment})
will divide both $c$ and~$d$.
When constructing an EtS scheme from~$F$,
because the compression function only takes fixed-size input,
we need to map the $(\ad,m)$ input to a series of block--tweak pairs $(B,t)$.
We will refer to this mapping as the input \defineterm{encoding}.

The approach taken by \cite{ESORICS:PijPoe20} fixes the encoding independently of the encryption engine.
It is precisely this modular approach that allows us to easily replace the encoding function with our optimized version.
We present our new encoding function in \secref{sec:padding} and provide the encryption engine in \secref{sec:encryption-construction} for completeness.
Together they form an efficient construction of EtS.

We first convey a rough overview of the EtS construction.
In \figref{fig:encrypttoself:example}
we consider an example with block size~$d$ double the chaining value size~$c$.
We assume that key~$k$ is padded to size~$d$.
The first block $B_1$ only contains associated data and we XOR $B_1$ with the key $k$ before we feed it into the compression function.
From the second block, we start processing message data.
We fill the first half of the block with associated data~$\ad_3$ and the second half with message data~$m_1$, and XOR with the key.
We also XOR $m_1$ with the current chaining value~$C_1$, to generate a partial ciphertext~$\ct_1$.
The same happens in the third block and we append~$\ct_2$ to the ciphertext.
If there is associated data left after processing all message data we can load the entire block with associated data, which occurs in the fourth block.
Note, we no longer need to XOR the key into the block after we have processed all message data, because at this point the input to the compression function will already be independent of the message $m$.
After processing all blocks, we XOR an offset $\omega\in\{\omega_0,\omega_1\}$ with the chaining value, where $\omega_0,\omega_1$ are two distinct constants.
The binding tag will be (a truncation of) the last chaining value~$C^*$.
\unskip\footnote{\ignorespaces
  It will be crucial to fix $\omega_0,\omega_1$ such that they are distinct also after truncation.
}
Note that the task of the encoding is not only to partition $\ad$ and~$m$ into blocks $B_1,B_2,\ldots$ as described,
but also to derive tweak values $t_1,t_2,\ldots$ and the choice of the final offset~$\omega$
in such a way that the overall encoding is injective.

\begin{figure*}[t]
  \centering
  \includegraphics{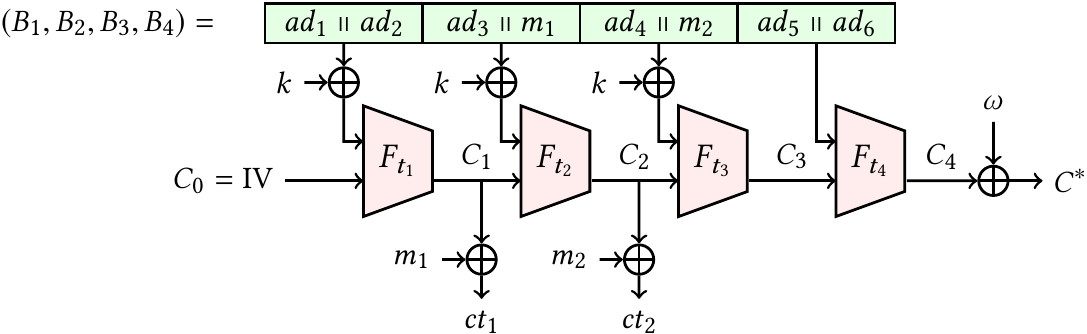}
  \caption[Example EtS operation]{\ignorespaces
      Example for $\enc(k,\ad,m)$ where $d=2c$ and $\ad=\ad_1\append\ldots\append\ad_6$ and $m=m_1\append m_2$ with $\abs{\ad}=6c$ and $\abs{m}=2c$.
      For clarity we have made the blocks~$B_i$, as they are output by the encoding function, explicit.
  \unskip}
  \label{fig:encrypttoself:example}
\end{figure*}

\subsection{Old and New Message Block Encoding}
\label{sec:padding}
We turn to the technical component of our EtS construction that encodes the $(\ad,m)$ input into a series of output pairs $(B,t)$ and the final offset value~$\omega$.
For authenticity we require that the encoding is injective.
For efficiency we require
that the encoding is online
(i.e., the input is read only once, left-to-right, and with small state),
that the number of output pairs is as small as possible,
and that the encoding preserves memory alignment
(see \secref{sec:memoryalignment}).
Syntactically, for the outputs we require that all $B\in\Sigma^d$, all $t\in T$, and $\omega\in\Omega$,
where quantities $c,d$ are those of the employed compression function, $T=\bits$, and $\Omega\subseteq\Sigma^c$ is any two-element set.
In our implementations we use $\Omega=\{\omega_0,\omega_1\}$ where $\omega_0=\texttt{0x00}^c$ and $\omega_1=\texttt{0xa5}^c$.
Pijnenburg {\etal} \cite{ESORICS:PijPoe20} describe the task as follows:

\medskip\noindent\textbf{Task.}
Assume $\abs\Sigma=256$ and $\adsp=\msgsp=\Sigma^*$ and $T=\bits$ and $\abs\Omega=2$.
For $c,d\in\NN^+$, $c<d$, find an injective encoding function $\encode\colon\adsp\times\msgsp\to(\Sigma^d\times T)^*\times\Omega$ that takes as input two finite strings and outputs a finite sequence of pairs $(B,t)\in\Sigma^d\times T$ and an offset $\omega\in\Omega$.
\medskip

As we already alluded to in \secref{sec:introduction} we introduce two improvements to the encoding scheme of~\cite{ESORICS:PijPoe20}.
These are not related to any aspect of security,
but rather to efficiency.
We note that implementing the encoding presented in \cite{ESORICS:PijPoe20} will require shifting every message byte in computer memory by one position.
As described in \secref{sec:memoryalignment} this shift-by-one operation is considerably more expensive than one might expect at first.
We alleviate this efficiency bottleneck by changing the padding to one which does not require shifting by a single byte (yet remains injective).

We will now present a detailed specification of our encoding (and decoding) function.
The pseudocode can be found in \figref{fig:encrypttoself:encoencr}, but we present it here in text.
Note the construction does not use the decoding function, but we provide it anyway to show that the encoding function is indeed injective.
Roughly, we encode as follows.
We fill the first block with associated data and for any subsequent block we load the associated data in the first part of the block and the message in the second part of the block.
When we have processed all the message data, we load the full block with $\ad$ again.
Clearly, we need to pad $\ad$ if it runs out before we have processed all message data.
We do this by appending a special termination symbol $\diamond\in\Sigma$ to~$\ad$ and then appending null bytes as needed.
Similarly, we need to pad the message if the message length is not a multiple of $c$.
Naturally, one might want to pad the message to a multiple of $c$.
However, this is suboptimal:
Consider the scenario where there are $d-c+1$ bytes remaining to be processed of associated data and 1 byte of message data.
In principle, message and associated data would fit into a single block,
but this would not be the case any longer if the message is padded to size~$c$.
On the other hand, for efficiency reasons we do not want to misalign all our remaining associated data.
If we do not pad at all, when we process the next $d$ bytes of associated data, we can only fit $d-1$ bytes in the block and have to put 1 byte into the next block.
Therefore, we pad~$m$ up to a multiple of the memory alignment value $\mav$.
To be precise, we pad message with null bytes until reaching a multiple of~$\mav$.
We replace the final (null) byte with the message length~$\abs{m}$;
this will uniquely determine where $m$ stops and the padding begins.
This restricts us to $c\leq256$ bytes such that $\abs{m}$ always can be encoded into a single byte.
As far as we are aware, any current practical compression function satisfies this requirement.

In \figref{fig:encoding:example},
for the artificially small case with $c=1$ and $d=2$
we provide four examples of what the blocks would look like for different inputs.
The top row shows the encoding of~8 bytes of associated data and an empty message.
The second row shows the encoding of empty associated data and 3~bytes of message data.
The third row shows the encoding of~6 bytes of associated data and 2~bytes of message data.
The final row shows the encoding of~3 bytes of associated data and 3~bytes of message data.

\begin{figure}[htb]
  \centering
  \includegraphics{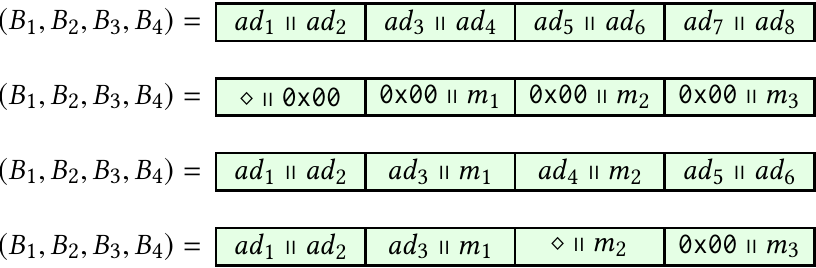}
  \caption[Encoding Example]{\ignorespaces
    Example encodings for the case $c=1$ and $d=2$.
  \unskip}
  \label{fig:encoding:example}
\end{figure}

We have two ambiguities remaining.
(1)~How to tell whether $\ad$ was padded or not?
Consider the first row in \figref{fig:encoding:example}.
What distinguishes the case $\ad=\ad_1\append\ldots\append\ad_7$ from $\ad=\ad_1\append\ldots\append\ad_7\append\ad_8$ with $\ad_8=\diamond$?
A similar question applies to the message.
(2)~How to tell whether a block contains message data or not?
Compare e.g., the first row with the third row.
This is where the tweaks come into play.

First of all, we tweak the first block if and only if the message is empty.
This fully separates the authentication-only case from the case where we have message input.

Next, if the message is non-empty, we use the tweaks to indicate when we switch from processing message data to $\ad$-only:
we tweak when we have consumed all of $m$, but still have $\ad$ left.
Note the first block never processes message data,
so the earliest block this may tweak is the second block and hence this rule does not interfere with the first rule.
Furthermore, observe this rule never tweaks the final block, as by definition of being the final block, we do not have any associated data left to process.

Next, we need to distinguish between the cases whether $m$ is padded or not.
In fact,
as the empty message was already taken care of,
we need to do this only if $m$ is at least one byte in size.
As in this case the final block does not coincide with the first block,
we can exploit that its tweak is still unused;
we correspondingly tweak the final block if and only if $m$ is padded.
Obviously, this does not interfere with the previous rules.

Finally, we need to decide whether $\ad$ was padded or not.
We do not want to enforce a policy of `always pad', as this could result in an extra block and hence an extra compression function invocation.
Instead, we use our offset output.
We set the offset~$\omega$ to~$\omega_1$ if $\ad$ was padded;
otherwise we set it to~$\omega_0$.

This completes our description of the encoding function.
The decoding function is a technical exercise carefully unwinding the steps taken in the encoding function, which we perform in \figref{fig:encrypttoself:encoencr}.
We obtain that for all $m\in\msgsp,\ad\in\adsp$ we have $\decode(\encode(\ad,m))=(\ad,m)$.
It immediately follows that our encoding function is injective.
For readability we have implemented the core functionality of the encoding in a coroutine called $\NEXT$, rather than a subroutine.
Instead of generating the entire sequence of $(B,t)$ pairs and returning the result, it will `Yield' one pair and suspend its execution.
The next time it is called (e.g., the next step in a for loop), it will resume execution from where it called `Yield', instead of at the beginning of the function, with all of its state intact.
The $\encode$ procedure is a simple wrapper that runs the $\NEXT$ procedure and collects its output, but our authenticated encryption engine described in \secref{sec:encryption-construction} will call the $\NEXT$ procedure directly.

\begin{figure*}[ht]
  \centering
  {\scalebox{1.045}{\includegraphics{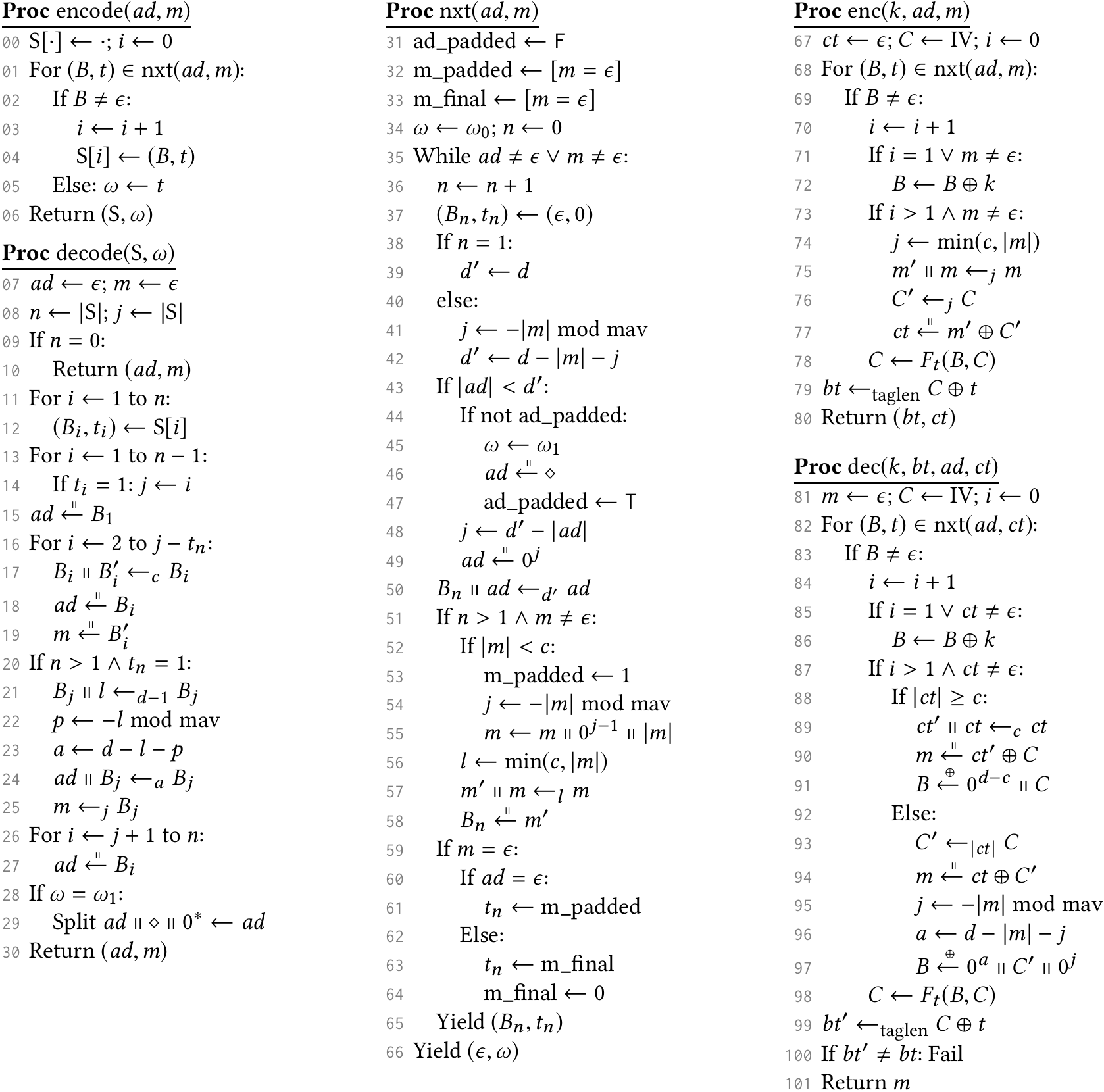}}}
  \caption[EtS: encoding and encryption]{\ignorespaces
    EtS construction: encoder, decoder, encryptor, and decryptor.
    (Procedure $\NEXT$ is a coroutine for $\encode$, $\enc$, and $\dec$, see text.)
    Using global constants $\mav$, $c$, $d$, $\taglen$, and $\IV$.
  \unskip}
  \label{fig:encrypttoself:encoencr}
\end{figure*}

\subsection{Encryption Engine}
\label{sec:encryption-construction}
For completeness, in this section we describe the encryption engine presented in~\cite{ESORICS:PijPoe20}.
Thanks to their modular approach, we can combine it with our encoding function without having to make any modifications:
They only assume that the associated data and message are present in encoded format, i.e.,
as a sequence of pairs $(B,t)$,
where $B\in\Sigma^d$ is a block and $t\in\bits$ is a tweak,
and additionally an offset $\omega\in\{\omega_0,\omega_1\}$.

We specify the encryption and decryption algorithms in \figref{fig:encrypttoself:encoencr} and assume they are provided with a key of length~$d$.
This is to ensure the $\oplus$ operation is well defined, without cluttering the notation.
One can consider the key padded with null bytes to length~$d$.
In practice, our implementation in C code will XOR the key with the first $\abs{k}$ bytes of a block $B$.
As described in \secref{sec:introduction}, the associated data string is often shorter than the message input.
By our encoding function defined in \secref{sec:padding}, this means that the $\ad$-input of most compression function invocations will be constant.
Thus, the first part of a block can be precomputed and no XORing is necessary any more.
Hence improving on the overall execution time compared to \cite{ESORICS:PijPoe20}.

We now discuss the encryption and decryption procedure presented in \cite{ESORICS:PijPoe20} in more detail.
As illustrated in \figref{fig:encrypttoself:example},
the main idea is to XOR the key with all blocks that are involved with message processing.
For the skeleton of the construction, we initialize the chaining value $C$ to $\IV$ and loop through the sequence of pairs $(B,t)$ output by the encoding function, each iteration updating the chaining value ${C\gets F_t(B,C)}$.
Let us examine each iteration of the $\enc$ procedure in more detail.
If the block is empty
(line~69),
we are in the final iteration and do not do anything.
Otherwise, we check if we are in the first iteration or if we have message data left
(line~71).
In this case we XOR the key into the block
(line~72).
This ensures we start with an unknown input block and that subsequent inputs are statistically independent of the message block.
If we only have $\ad$ remaining we can use the block directly as input to the compression function.
If we have message data left we will encrypt it starting from the second block
(line~73).
To encrypt, we take a chunk of the message, XOR it with the chaining value of equal size and append the result to the ciphertext
(lines~74--77).
We only start encrypting from the second iteration as the first chaining value is public.
Finally, we call the compression function $F_t$ to update our chaining value
(line~78).
Once we have finished the loop, the last pair $(B,t)$ equals $(\es,\omega)$ by definition.
So we XOR the offset~$\omega$ with the chaining value $C$ and truncate the result to obtain the binding tag
(line~79).
We return the binding tag along with the ciphertext.

The $\dec$ procedure is similar to the $\enc$ procedure but needs to be slightly adapted.
Informally, the $\NEXT$ procedure now outputs a block $B=(\ad\append\ct)$
(line~82)
instead of $B=(\ad\append m)$
(line~68).
Hence, we XOR with the chaining variable
(line~91,97)
such that the block becomes $B=(\ad\append m)$ and the compression function call takes equal input compared to the $\enc$ procedure.
The case distinction handles the slightly different positioning of ciphertext in the blocks.
Finally, there obviously is a check if the computed binding tag is equal to the stored binding tag
(line~100).

\subsection{Security Analysis}
In order to prove security, we need further assumptions on our compression function than the standard assumption of preimage resistance and collision resistance.
For example, we need $F$ to be difference unpredictable.
Roughly, this notion says it is hard to find a pair $(x,y)$ such that $F(x)=F(y)\oplus z$ for a given difference $z$.
Moreover, we truncate the binding tag, so actually it should be hard to find a tuple such that this equation holds for the first $\abs{\bt}$ bits.
We note collision resistance of~$F$ does not imply collision resistance of a truncated version of $F$ \cite{C:BihChe04}.
However, such assumptions could be justified when one considers the compression function as a random function.
Hence, instead of several ad~hoc assumptions, we prove our construction secure directly in the random oracle model.

As described in \cite{ESORICS:PijPoe20},
the SHA2 compression function can be tweaked by modifying the chaining value depending on the tweak.
Let $F$ be the tweakable compression function in \figref{fig:encrypttoself:encoencr}.
We write $F'$ for the SHA2 compression function that will take as input the block and the (modified) chaining value.
Let $H\colon\Sigma^d\times\Sigma^c\to\Sigma^c$ be a random oracle.
In the security analysis of the SHA2 construction, we will substitute $H$ for $F'$ in our construction.

We remark the BLAKE2b compression function is a tweakable compression function and it can be substituted directly for a random oracle with an extended input space.
That is, a random oracle $\bar{H}\colon\Sigma^d\times\bits\times\Sigma^c\to\Sigma^c$.
Hence, in the security analysis of the BLAKE2b construction, we will substitute $\bar{H}$ for $F$ in our construction.

We remark that we cannot treat our tweaked SHA2 compression function $F$ in this way as it would be distinguishable from random oracle $\bar{H}$.
To see this, observe that querying $F$ on the unmodified chaining variable with tweak $t=1$ yields the same result as querying $F$ on the modified chaining variable with $t=0$.
In the random oracle $\bar{H}$ these two queries are completely independent.

Both for tweakable and non-tweakable compression functions
our EtS construction from \figref{fig:encrypttoself:encoencr} provides
integrity
and indistinguishability
in the random oracle model,
assuming sufficiently large tag and key lengths.
We refer the reader to Proposition~\ref{thm:ets:int} for the integrity, and Proposition~\ref{thm:ets:ind} for the indistinguishability, of the instantiation with a non-tweakable compression function.
For the instantiation with a tweakable compression function we refer to Proposition~\ref{thm:ets:int2} and Proposition~\ref{thm:ets:ind2} for integrity and indistinguishability, respectively.

The security proofs in \cite{ESORICS:PijPoe20} only require injectivity from the encoding function.
We have demonstrated in \secref{sec:padding} that our modified encoding function remains injective, so the security proofs still apply.
However, \cite{ESORICS:PijPoe20} omits the security proofs for the tweakable instantiation, so we expand upon them here for the interested reader.
We will now first discuss the non-tweakable compression function instantiation and subsequently the tweakable compression function instantiation.

Let $H\colon\Sigma^d\times\Sigma^c\to\Sigma^c$ be a random oracle.
Recall we consider an instantiation with a standard (non-tweakable) compression function $F'$ transformed into a tweakable compression function $F$ by modifying the chaining value.
We replace $F'$, used internally by $F$, with random oracle $H$.

\begin{proposition}
    \label{thm:ets:int}
    Let $\pi$ be the construction given in \figref{fig:encrypttoself:encoencr},
    $H$ a random oracle replacing the compression function,
    $\advA$ an adversary,
    \allowbreak
    $\Adv_\pi^\int(\advA)$ the advantage that $\advA$ has against $\pi$ in the integrity game of \figref{fig:encrypttoself:games}
    and $q$ the number of random oracle queries (either directly or indirectly via $\Odec$).
    We have,
    \begin{align*}
        \Adv_\pi^\int(\advA)\leq q^2\cdot2^{-c}+q\cdot2^{-\abs{\bt}}.
    \end{align*}
\end{proposition}

\begin{proof}
    For all $\ad\in\adsp,m\in\msgsp$
    we will show that
    \begin{align*}
        \Pr[\INT(\ad,m,\advA)]\leq q^2\cdot2^{-c}+q\cdot2^{-\abs{\bt}}.
    \end{align*}
    Let $\ad\in\adsp$ be associated data and let $m\in\msgsp$ be a message.
    The game $\INT(\ad,m,\advA)$ samples a uniformly random key $k\in\keysp$ and computes $(\bt,c)=\enc(k,\ad,m)$.
    $\advA$ wins the $\INT$ game if it provides a pair $(\ad',c')\neq(\ad,c)$ such that $\dec(k,\bt,ad',c')$ succeeds, which only happens if $\bt'=\bt$.
    Recall the encoding function outputs a sequence $S$ of $(B,t)$ pairs and an offset $\omega$.
    Because the encoding function is injective we must have $S'\neq S$ or $\omega'\neq\omega$.
    Let us first assume $S'=S$.
    Let $C_n$ denote the final chaining variable.
    Because the sequences are equal, we will arrive at $C'_n=C_n$.
    We must have $\omega'\neq\omega$, but clearly $C_n\oplus\omega_0$ is not equal to $C_n\oplus\omega_1$ (even after truncation),
    that is, $\bt'\neq\bt$.
    We have a contradiction and conclude $S'\neq S$.

    For the case $S'\neq S$, let us now assume the subcase $\omega'\neq\omega$.
    The first $\abs{\bt}$ bits of $C'_{n'}$ must equal the first $\abs{\bt}$ bits of $C_n\oplus\omega\oplus\omega'$, i.e., $\advA$ must find a partial preimage.
    Because $H$ is a random oracle, $\advA$ would succeed with probability at most $q\cdot2^{-\abs{\bt}}$, where $q$ is the number of queries.
    In the other subcase we have $\omega'=\omega$.
    Then the first $\abs{\bt}$ bits of $C'_{n'}$ must equal the first $\abs{\bt}$ bits of $C_n$, i.e.,
    the first $\abs{\bt}$ bits of $H(B'_{n'},\hat{C}'_{n'-1})$ must equal the first $\abs{\bt}$ bits of $H(B_n,\hat{C}_{n-1})$,
    where $\hat{C}'_{n'-1},\hat{C}_{n-1}$ are the chaining values $C'_{n'-1},C_{n-1}$ after applying tweaks $t'_{n'},t_n$, respectively.
    If the inputs are not equal, $\advA$ has found a partial second preimage.
    Since $H$ is a random oracle, $\advA$ would succeed with probability at most $q\cdot2^{-\abs{\bt}}$, where $q$ is the number of oracle queries.
    However, if the inputs are equal we know $\hat{C}'_{n'-1}=\hat{C}_{n-1}$.
    Let us write $\hat{C}'_{n'-1}=C'_{n'-1}\oplus\tau'$ and $\hat{C}_{n-1}=C_{n-1}\oplus\tau$.
    We obtain $C'_{n'-1}=C_{n-1}\oplus\tau\oplus\tau'$.
    Thus, either $\advA$ has found a collision or $C'_{n'-1}=C_{n-1}$.
    We can repeat the argument to reason about $C'_{n-2},C_{n-2}$, etc.
    By a standard birthday argument we can bound the probability of a collision by $q^2\cdot2^{-c}$.

    If we eventually conclude $C'_{n'-\delta}=C_{n-\delta}=\IV$, we know one of the sequences is longer, i.e., $n'-\delta>0$ or $n-\delta>0$.
    Otherwise the sequences would be equal, which is excluded by the injectivity of the encoding function.
    In the case $n-\delta>0$, there has been a collision in the hash function, we have already bounded this probability above.
    Thus, let us assume $n'-\delta>0$.
    We have $H(B_{n'-\delta},\hat{C}_{n'-\delta-1})=\IV$.
    Thus $\advA$ has found a preimage of $\IV$.
    Because $H$ is a random oracle, $\advA$ would succeed with probability at most~$q\cdot2^{-c}$.
\end{proof}

\begin{proposition}
    \label{thm:ets:ind}
    Let $\pi$ be the construction given in \figref{fig:encrypttoself:encoencr},
    $H$ a random oracle replacing the compression function,
    $\advA$ an adversary,
    \allowbreak
    $\Adv_\pi^\ind(\advA)$ the advantage that $\advA$ has against $\pi$ in the indistinguishability games of \figref{fig:encrypttoself:games}
    and $q$ the number of random oracle queries (either directly or indirectly via $\Odec$).
    We have,
    \begin{align*}
        \Adv_\pi^\ind(\advA)\leq q^2\cdot2^{-c}+q\cdot2^{-\abs{k}}+\Adv_\pi^\int(\advA).
    \end{align*}
\end{proposition}

\begin{proof}
    Other than the challenge pair $(\ad,c)$, we can assume the decryption oracle rejects all queries by~$\advA$.
    Otherwise $\advA$ would immediately win the integrity game and the proposition holds.
    Encryption is done by XORing the message with the chaining variable.
    As long as the chaining variable never repeats, each input to $H$ is a fresh query that has not been seen before.
    Then $H$ will provide fresh, uniformly random output, as it is a random oracle.
    By a standard birthday argument we can bound the probability of a collision by $q^2\cdot2^{-c}$.
    Now let us assume there is no collision.
    Each chaining variable that is used to encrypt is output of a query to $H$ that XORed the key $k$ with the input.
    Additionally each block that has message data as input is also XORed with the key $k$.
    Thus if $\advA$ does not know $k$ it cannot query $H$ to obtain the chaining variable.
    The key is only used with input to the compression function,
    and since $H$ is a random oracle,
    $\advA$ can only learn by guessing the input and checking the random oracle output.
    However, this has a success probability of at most $q\cdot2^{-\abs{k}}$.
\end{proof}

Let $\bar{H}\colon\Sigma^d\times\bits\times\Sigma^c\to\Sigma^c$ be a random oracle.
We now consider an instantiation with a tweakable compression function $F$.
We replace $F$ with random oracle $\bar{H}$.
\begin{proposition}
    \label{thm:ets:int2}
    Let $\pi$ be the construction given in \figref{fig:encrypttoself:encoencr},
    $\bar{H}$ a random oracle replacing the tweakable compression function,
    $\advA$ an adversary,
    \allowbreak
    $\Adv_\pi^\int(\advA)$ the advantage that $\advA$ has against $\pi$ in the integrity game of \figref{fig:encrypttoself:games}
    and $q$ the number of random oracle queries (either directly or indirectly via $\Odec$).
    We have,
    \begin{align*}
        \Adv_\pi^\int(\advA)\leq q^2\cdot2^{-c}+q\cdot2^{-\abs{\bt}}.
    \end{align*}
\end{proposition}
\begin{proof}
    For all $\ad\in\adsp,m\in\msgsp$
    we will show that
    \begin{align*}
        \Pr[\INT(\ad,m,\advA)]\leq q^2\cdot2^{-c}+q\cdot2^{-\abs{\bt}}.
    \end{align*}
    Let $\ad\in\adsp$ be associated data and let $m\in\msgsp$ be a message.
    The game $\INT(\ad,m,\advA)$ samples a uniformly random key $k\in\keysp$ and computes $(\bt,c)=\enc(k,\ad,m)$.
    $\advA$ wins the $\INT$ game if it provides a pair $(\ad',c')\neq(\ad,c)$ such that $\dec(k,\bt,ad',c')$ succeeds, which only happens if $\bt'=\bt$.
    Recall the encoding function outputs a sequence $S$ of $(B,t)$ pairs and an offset $\omega$.
    Because the encoding function is injective we must have $S'\neq S$ or $\omega'\neq\omega$.
    Let us first assume $S'=S$.
    Let $C_n$ denote the final chaining variable.
    Because the sequences are equal, we will arrive at $C'_n=C_n$.
    We must have $\omega'\neq\omega$, but clearly $C_n\oplus\omega_0$ is not equal to $C_n\oplus\omega_1$ (even after truncation),
    that is, $\bt'\neq\bt$.
    We have a contradiction and conclude $S'\neq S$.

    For the case $S'\neq S$, let us now assume the subcase $\omega'\neq\omega$.
    The first $\abs{\bt}$ bits of $C'_{n'}$ must equal the first $\abs{\bt}$ bits of $C_n\oplus\omega\oplus\omega'$, i.e., $\advA$ must find a partial preimage.
    Because $\bar{H}$ is a random oracle, $\advA$ would succeed with probability at most $q\cdot2^{-\abs{\bt}}$, where $q$ is the number of queries.
    In the other subcase we have $\omega'=\omega$.
    Then the first $\abs{\bt}$ bits of $C'_{n'}$ must equal the first $\abs{\bt}$ bits of $C_n$, i.e.,
    the first $\abs{\bt}$ bits of $\bar{H}(B'_{n'},t_{n'},C'_{n'-1})$ must equal the first $\abs{\bt}$ bits of $\bar{H}(B_n,t_n,C_{n-1})$.
    If the inputs are not equal, $\advA$ has found a partial second preimage.
    Since $\bar{H}$ is a random oracle, $\advA$ would succeed with probability at most $q\cdot2^{-\abs{\bt}}$, where $q$ is the number of oracle queries.
    However, if the inputs are equal we know $C'_{n'-1}=C_{n-1}$.
    Thus, either $\advA$ has found a collision or $C'_{n'-1}=C_{n-1}$.
    We can repeat the argument to reason about $C'_{n-2},C_{n-2}$, etc.
    By a standard birthday argument we can bound the probability of a collision by $q^2\cdot2^{-c}$.

    If we eventually conclude $C'_{n'-\delta}=C_{n-\delta}=\IV$, we know one of the sequences is longer, i.e., $n'-\delta>0$ or $n-\delta>0$.
    Otherwise the sequences would be equal, which is excluded by the injectivity of the encoding function.
    In the case $n-\delta>0$, there has been a collision in the hash function, we have already bounded this probability above.
    Thus, let us assume $n'-\delta>0$.
    We have $\bar{H}(B_{n'-\delta},t_{n'-\delta},C_{n'-\delta-1})=\IV$.
    Thus $\advA$ has found a preimage of $\IV$.
    Because $\bar{H}$ is a random oracle, $\advA$ would succeed with probability at most~$q\cdot2^{-c}$.
\end{proof}

\begin{proposition}
    \label{thm:ets:ind2}
    Let $\pi$ be the construction given in \figref{fig:encrypttoself:encoencr},
    $\bar{H}$ a random oracle replacing the tweakable compression function,
    $\advA$ an adversary,
    \allowbreak
    $\Adv_\pi^\ind(\advA)$ the advantage that $\advA$ has against $\pi$ in the indistinguishability games of \figref{fig:encrypttoself:games}
    and $q$ the number of random oracle queries (either directly or indirectly via $\Odec$).
    We have,
    \begin{align*}
        \Adv_\pi^\ind(\advA)\leq q^2\cdot2^{-c}+q\cdot2^{-\abs{k}}+\Adv_\pi^\int(\advA).
    \end{align*}
\end{proposition}

\begin{proof}
    Other than the challenge pair $(\ad,c)$, we can assume the decryption oracle rejects all queries by~$\advA$.
    Otherwise $\advA$ would immediately win the integrity game and the proposition holds.
    Encryption is done by XORing the message with the chaining variable.
    As long as the chaining variable never repeats, each input to $\bar{H}$ is a fresh query that has not been seen before.
    Then $\bar{H}$ will provide fresh, uniformly random output, as it is a random oracle.
    By a standard birthday argument we can bound the probability of a collision by $q^2\cdot2^{-c}$.
    Now let us assume there is no collision.
    Each chaining variable that is used to encrypt is output of a query to $\bar{H}$ that XORed the key $k$ with the input.
    Additionally each block that has message data as input is also XORed with the key $k$.
    Thus if $\advA$ does not know $k$ it cannot query $\bar{H}$ to obtain the chaining variable.
    The key is only used with input to the compression function,
    and since $\bar{H}$ is a random oracle,
    $\advA$ can only learn by guessing the input and checking the random oracle output.
    However, this has a success probability of at most $q\cdot2^{-\abs{k}}$.
\end{proof}

\section{Implementation}

We implemented three versions of the EtS primitive.
We developed optimized C~code for the padding scheme and encryption engine from \figref{fig:encrypttoself:encoencr},
based on the compression functions of common hash functions.
Specifically, our EtS implementations are based on the compression functions of SHA256, SHA512, and BLAKE2 \cite{NIST:FIPS180-4,rfc7693}.
We chose these functions as all three of them are ARX designs (Add--Rotate--Xor) which makes them particularly efficient in software implementations.
While SHA256 and SHA512 are more widely standardized and used than BLAKE2,
only the latter is a HAIFA construction and tweakable without ad-hoc modifications.
Note that due to the used internal register size of 32~bits,
SHA256 is most competitive on 32-bit CPUs;
in contrast,
SHA512 and BLAKE2 use 64-bit registers and thus perform best on 64-bit CPUs.

We implemented all components of EtS in plain~C,
including the compression functions, the encoding schemes, and the EtS framework.
In addition we implemented a range of self-tests and provide test vectors.
We note that while in particular the compression functions would be good candidates for being re-implemented in assembly for further efficiency improvements,
we believe that, as all three compression functions are ARX designs, the penalty of not hand-optimizing is not too drastic.

We released the source code of our implementation as open source software.
The terms of use are those granted by the Apache license
\unskip\footnote{\url{https://www.apache.org/licenses/LICENSE-2.0}}.
The code is available at \url{https://github.com/cryptobertram/encrypt-to-self}.

We conducted timing measurements for our implementations.
We measured on two devices:
on a roughly 9-year old CPU that identifies itself as \texttt{Intel Core i3-2350M CPU @ 2.30GHz},
and on a more recent CPU of the type \texttt{Intel Core i5-7300U CPU @ 2.60GHz}.
The results are shown in Table~\ref{tab:timings}.
The timings were taken for various message lengths, with a 16~byte associated data input in call cases.
Note that the BLAKE2 based version clearly outperforms the others for all tested message lengths.
Further, SHA512 is generally faster than SHA256
(except for messages that are so short that one SHA256 compression function invocation is sufficient to fully encrypt the message).

\begin{table}[ht]
  \centering
  \caption{Timings (in microseconds) of EtS implementation}
  \includegraphics{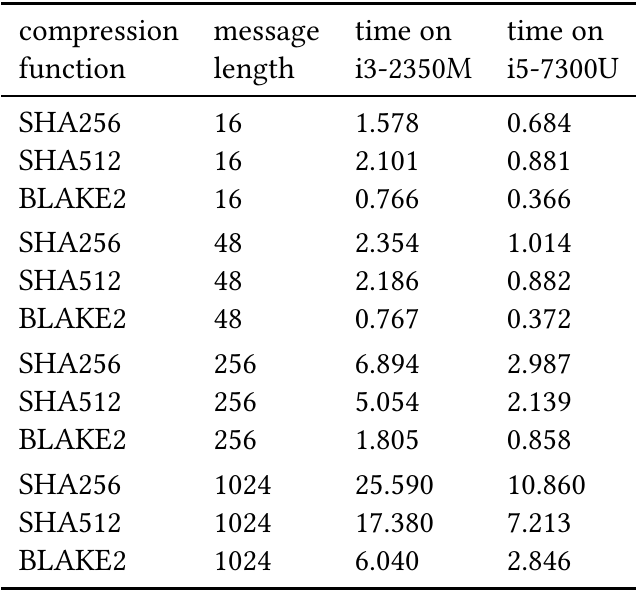}
  \label{tab:timings}
\end{table}

\begin{acks}
  We thank the reviewers of CYSARM'20 for their helpful comments and feedback.
  The research of Pijnenburg was supported by the \grantsponsor{EPSRC}{EPSRC and the UK government}{https://gow.epsrc.ukri.org/NGBOViewGrant.aspx?GrantRef=EP/P009301/1} as part of the Centre for Doctoral Training in Cyber Security at Royal Holloway, University of London (\grantnum{EPSRC}{EP/P009301/1}).
  The research of Poettering was supported by the \grantsponsor{Horizon}{European Union's Horizon~2020 project FutureTPM}{https://cordis.europa.eu/project/id/779391} (\grantnum{Horizon}{779391}).
\end{acks}

\end{document}